\documentclass[11pt]{article}

\usepackage[margin=1.1in]{geometry}
\usepackage{amsmath,amssymb,amsthm}
\usepackage{microtype}
\usepackage{graphicx}
\usepackage[round]{natbib}
\usepackage{titlesec}
\usepackage{hyperref}
\usepackage{xurl}
\hypersetup{colorlinks=true,linkcolor=black,citecolor=black,urlcolor=blue,
  pdftitle={Marginally Useful},pdfauthor={Peter Cotton}}

\titleformat{\paragraph}[runin]{\normalfont\itshape}{}{0pt}{}
\titlespacing*{\paragraph}{0pt}{0.6\baselineskip}{0.6em}

\newtheorem{proposition}{Proposition}
\theoremstyle{definition}
\newtheorem{remark}{Remark}

\newcommand{\E}{\mathbb{E}}
\renewcommand{\Pr}{\mathbb{P}}
\newcommand{\KL}{\mathrm{KL}}
\newcommand{\given}{\,\vert\,}

\title{Marginally Useful:\\ An Information-Gap Identity in Split Conformal Prediction}

\author{Peter Cotton\thanks{Global Strategic Minerals Corporation.
  Email: \texttt{peter.cotton@gsmc.ai}. Code, figures, and an interactive companion are at
  \url{https://conformalprediction.net} and
  \url{https://github.com/microprediction/conformalprediction}.}}
\date{}

\begin{document}
\maketitle

\begin{abstract}
\noindent
Conformal prediction has been touted as a more formal, rigorous approach to adding
uncertainty to a forecast. The sole objective of this note is to point out that rigor cuts
both ways in the case of residual pooling, the technique used in the vast majority of
conformal prediction applications. The fact that unconditional guarantee of coverage is
provided is not in question, but we make clear, we believe for the first time, that there is
an opposing guarantee too: a permanent gambit of logarithmic-score regret which no amount of
data or tuning can subsequently reduce. We give the exact size of the sacrifice, and a
financial reading of it as the growth rate of an oracle adversary betting against odds set
by someone using residual pooling.
\end{abstract}

\section{The claim}

In recent years conformal prediction has gained attention in the machine learning community
in particular. However, there is a danger this emphasis becomes a substitute for careful
statistical modeling. The use of split conformal prediction, where the field started, is
essentially synonymous with the statistical advice that says one's model residuals are best
modeled by their past empirical distribution. Such advice is arguably naive in fields like
finance where stochastic volatility is a mainstay. On the other hand there are many
well-used algorithms in classification, to pick one area, where the off the shelf outputs
(``proba'') tend to be very poorly calibrated, and in those settings adding a conformal
wrapper has real utility.

Our point is not to persuade anyone against pragmatism. It is that there can be no free
lunch, and that there is nothing sacred about the empirical distribution however practical
it may prove as a first remediation. The purpose of this note is to convert that common
sense into a crisp mathematical statement.

It should go without saying that the split conformal guarantee is not in question. Given
exchangeable data, a fitted predictor, and a nonconformity score, split conformal prediction
returns a set $C_\alpha(x)$ with
\begin{equation}\label{eq:marginal}
  \Pr\!\big(Y_{n+1}\in C_\alpha(X_{n+1})\big)\;\ge\;1-\alpha,
\end{equation}
finite-sample and with no assumption on the data distribution
\citep{vovk2005algorithmic,lei2018distribution,angelopoulos2023gentle}. The danger lies
solely in the {\em interpretation} of \eqref{eq:marginal} when a user wants a conditional
distributional prediction of the next outcome. The averaging here takes place over both
$X_{n+1}$ and $Y_{n+1}$, and that is not at all the usual notion of prediction. The paper
\citet{lei2014distribution} is one of many warnings about this in the literature, where it
is stated that ``a good estimator must satisfy something more than marginal coverage''.

The claims nonetheless get stronger as one descends from the journals toward the
practitioner. One widely used forecasting library tabulates methods by ``Calibration
Guarantee'', awarding the check mark to conformal prediction and a tilde to bootstrap,
quantile regression and ARIMA alike \citep{nixtla_conformal}. A popular introduction adds
that with Bayesian posteriors or bootstrapping ``we have no guarantees that these approaches
will perform well on new data'' \citep{molnar_conformal}. From there it is a short step to
conformal prediction as the thing that ``converts any statistical, machine or deep learning
model into a probabilistic prediction model'' \citep{manokhin_neuralprophet}. That is the
step we think cannot be taken for free, and below we price it.

\section{Setup}

Let $(X_i,Y_i)_{i=1}^{n+1}$ be exchangeable. Fix a point predictor $\hat\mu$ trained on a
disjoint set, and a \emph{nonconformity score} $s(x,y)$. The canonical choice in regression
is the absolute residual $s(x,y)=|y-\hat\mu(x)|$. Compute calibration scores
$s_i=s(X_i,Y_i)$ for $i=1,\dots,n$ and let
\begin{equation}\label{eq:q}
  \hat q \;=\; s_{(k)}, \qquad k=\big\lceil (n+1)(1-\alpha)\big\rceil,
\end{equation}
the $k$-th order statistic (with $\hat q=+\infty$ when $k>n$). The split-conformal set is
\begin{equation}\label{eq:set}
  C_\alpha(x)=\{y:\ s(x,y)\le \hat q\}
  \;=\;[\hat\mu(x)-\hat q,\ \hat\mu(x)+\hat q]
\end{equation}
for the absolute-residual score. Exchangeability of the augmented scores makes the rank of
$s_{n+1}$ uniform, which yields \eqref{eq:marginal}
\citep{vovk2005algorithmic,lei2018distribution}. None of this is in dispute. The question is
what it is taken to mean.

The guarantee one can have distribution-free is the marginal one, and a marginally valid
band ``tends to overestimate the set when $x$ is in the high density area and to
underestimate for low density $x$'' \citep{lei2014distribution}. It over-covers the easy
inputs and under-covers the hard ones. That's the game, and Figure~\ref{fig:margcond} shows
it on heteroscedastic data: $90\%$ overall, near-$100\%$ where the noise is small, well
below target where it is large. Conditional coverage is not a repair one can buy, since a
band with non-trivial finite-sample conditional validity has infinite expected length at
almost every point of a continuous distribution \citep{lei2014distribution}, a result
refined by \citet{barber2021limits}.

\begin{figure}[t]
\centering
\includegraphics[width=0.88\textwidth]{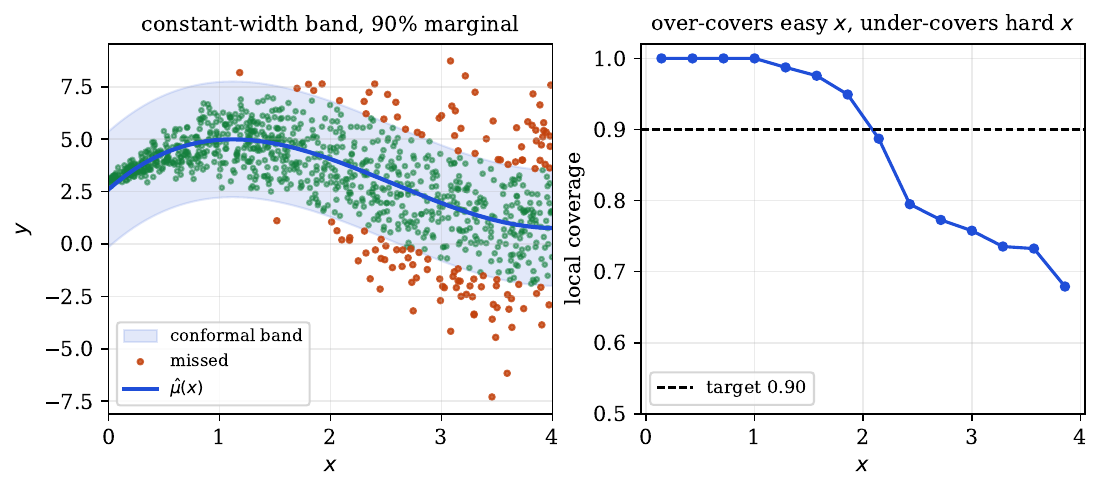}
\caption{Marginal is not conditional. A constant-width split-conformal band (left) attains
$90\%$ \emph{marginal} coverage on heteroscedastic data, but its \emph{local} coverage
(right) runs near $100\%$ where the noise is small and well below target where it is large.}
\label{fig:margcond}
\end{figure}

Marginal validity is not, on its own, evidence of a good method. We would argue that
coverage and log-likelihood belong on separate axes, as the two separate things a forecaster
reports. Those can be a point predictor $\hat\mu(x)$ and a full predictive density
$f(\cdot\given x)$. The canonical residual score $s(x,y)=|y-\hat\mu(x)|$ uses only
$\hat\mu$, the \emph{location} of the forecast. It never sees the spread or shape of $f$, so
the conformal set and its coverage are invariant to exactly the part of $f$ that the
log-score rewards.

\begin{remark}[Residual-score coverage does not constrain the log-score]\label{rem:orth}
Two forecasters with the \emph{same} location $\hat\mu$ but different predictive densities
$f$ have \emph{identical} conformal sets, and identical marginal coverage at every level
$\alpha$, while their expected log-scores $\E[\log f(Y\given X)]$ may differ by an arbitrary
amount. Take $X$ degenerate and $Y\sim N(0,1)$, so both use $s(x,y)=|y|$ and the conformal
set is the fixed interval $[-\hat q,\hat q]$ whatever density is claimed. With
$f_\sigma=N(0,\sigma^2)$, $\E[\log f_\sigma(Y)]=-\tfrac12\log(2\pi\sigma^2)-1/(2\sigma^2)$
diverges to $-\infty$ at both ends of the $\sigma$ range, while the set never moves.
\end{remark}

This is nothing more than a formal version of a point numerous authors have been making
informally. Perhaps the most blunt of these commentators is \citet{recht2024cover}, who
writes that conformal validity ``almost invites you to use garbage prediction functions''.
How could it be otherwise, when the spread never entered \eqref{eq:q} in the first place?

\section{The conformal information gap}

Our identities apply when split conformal prediction is made to output a distribution by means
of the \emph{conformal predictive system} (CPS) of \citet{vovk2019nonparametric}. We view this
 as a constrained attempt to fix the residuals {\em subject to the
restriction of using only one unconditional shared shape} as compared to a data-dependent
scale such as one would see in, say, a GARCH model. 

To wit, fix a location predictor $\hat\mu$ with residual $R=Y-\hat\mu(X)$ of marginal density $g$
and CDF $G$. In the large-calibration limit the CPS predictive distribution is
$\Pi_x(y)=G(y-\hat\mu(x))$, with density $\pi_x(y)=g(y-\hat\mu(x))$. This is the base
location forecast re-leveled by the marginal signed-residual law, and the estimation in it
was all done by the base model and the residual sample. Conformal supplies the leveling and
nothing else. 

One can write the true conditional residual density as $r(\cdot\given x)$, and the
oracle density as $q^\star(y\given x)=r(y-\hat\mu(x)\given x)$. We let $\bar r=g$ be the
marginal residual density.

\begin{proposition}[The residual-information gap]\label{prop:gap}
With $I(R;X)=\E_X\,\KL(r(\cdot\given X)\,\|\,\bar r)\ge 0$ the mutual information between
the residual and the input, and assuming the relevant densities and information quantities
are finite,
\begin{equation}\label{eq:gap}
  \E[\log q^\star(Y\given X)]-\E[\log \pi_X(Y)] \;=\; I(R;X).
\end{equation}
Among all single-shape forecasters $y\mapsto h(y-\hat\mu(x))$, $\E[\log h(R)]$ is maximized
at $h=\bar r$, and the signed CPS attains this optimum. Hence no recalibration that ignores
$X$ can reduce $I(R;X)$. Reducing it requires conditioning on $X$.
\end{proposition}

\begin{proof}
Let $\mathcal{R}(h)$ be the expected log-score regret of the single-shape forecaster
$y\mapsto h(y-\hat\mu(x))$ relative to the oracle. The oracle scores $Y$ by
$\log q^\star(Y\given X)=\log r(R\given X)$ and this forecaster by $\log h(R)$. The joint
$P_{X,R}$ has density $p(x)\,r(\rho\given x)$ and the product reference $P_X\otimes H$ has
density $p(x)\,h(\rho)$, so
\begin{equation}\label{eq:master}
  \mathcal{R}(h)=\E\!\left[\log\frac{r(R\given X)}{h(R)}\right]
      =\KL\!\big(P_{X,R}\,\big\|\,P_X\otimes H\big)
      = I(R;X)+\KL(\bar r\,\|\,h),
\end{equation}
the last step the chain rule for relative entropy along a product reference, the $R$-marginal
of $P_{X,R}$ being $\bar r$. Taking $h=\bar r$ kills the second term and gives
\eqref{eq:gap}. That second term is $\ge 0$ and is minimized at $h=\bar r$ by Gibbs, so no
$X$-blind shape improves on $I(R;X)$, and $I(R;X)=\KL(P_{X,R}\,\|\,P_X P_R)$ is reduced only
by conditioning on $X$.
\end{proof}

The right-hand side of
\eqref{eq:gap} is exactly the information about the residual distribution that remains in
$X$ after the location predictor is fixed, and calling it the conformal information gap
gives ``blind to conditional residual shape'' a number. 

The signed CPS is log-score optimal among single-shape location forecasters, so conformal
prediction is not dominated within its own class. What it pays for is the class restriction,
not the calibration step. Reading ordinary absolute-residual intervals as a distribution
costs a little more, but only a little.\footnote{The unsigned interval system is not a
predictive system, but sweeping $\alpha$ through the nested intervals
$\hat\mu(x)\pm\hat q_\alpha$ produces one implicitly. It is symmetric about $\hat\mu(x)$
with $|{\cdot}|$ distributed as $|R|$, so its density is the symmetrized marginal residual
law $h_{\mathrm{sym}}(z)=\tfrac12\big(g(z)+g(-z)\big)$, and \eqref{eq:master} with
$h=h_{\mathrm{sym}}$ gives a second term,
$\E[\log q^\star(Y\given X)]-\E[\log h_{\mathrm{sym}}(R)]=I(R;X)+\KL(\bar r\,\|\,
h_{\mathrm{sym}})$. That second term is the price of discarding the sign. It vanishes when
the residual law is symmetric, and since $h_{\mathrm{sym}}\ge\tfrac12\bar r$ it never
exceeds $\log 2$, one bit. The structural loss is $I(R;X)$ either way.}

\paragraph{A false-pooling cost.} A single-shape system asserts that, once $\hat\mu$ is
fixed, one residual law fits everyone, which is to say that $R\perp X$. The gap is the
relative entropy from the true residual experiment to the nearest such independence model,
and one pays it whether or not one notices making it.

\paragraph{A betting rent.} By the log-optimal (Kelly) growth theorem, a gambler who knows a
distribution $P$ and bets against a fair book priced by $Q$ compounds wealth at expected
log-rate $\KL(P\,\|\,Q)$ per round \citep{kelly1956new,cover2006elements}. Read $\bar r$ as
the book posted by a single-shape conformal forecaster and $r(\cdot\given X)$ as what the
oracle knows. The oracle's expected rent per round is
$\E_X\,\KL(r(\cdot\given X)\,\|\,\bar r)=I(R;X)$. A re-leveling that ignores $X$ leaves this
divergence unchanged, and so leaves the growth rate unchanged.

\begin{figure}[t]
\centering
\includegraphics[width=0.88\textwidth]{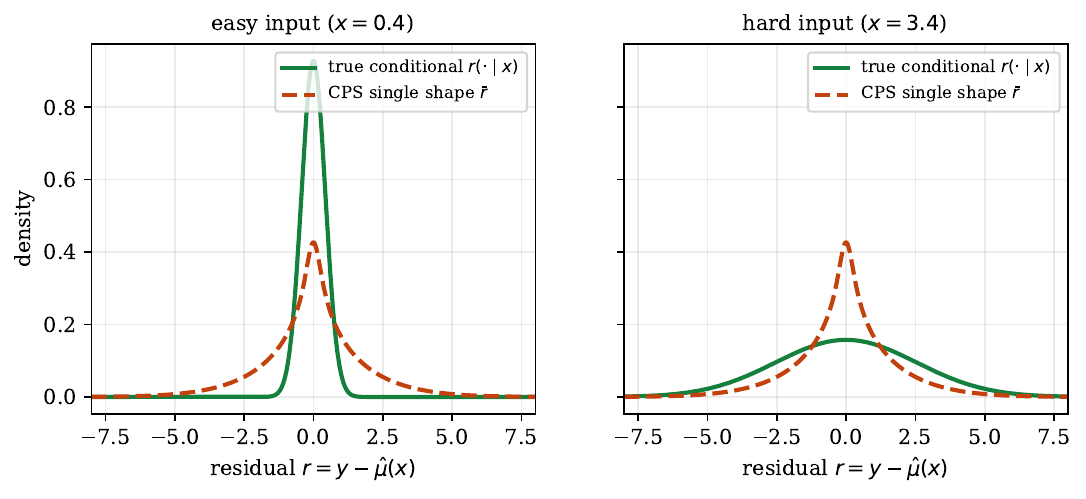}
\caption{The gap. A single-shape conformal predictive system uses the marginal residual law
$\bar r$ for every input. The oracle uses the true conditional $r(\cdot\given x)$, narrow
for easy inputs (left) and wide for hard ones (right). The expected log-score regret is
exactly $I(R;X)$.}
\label{fig:gap}
\end{figure}

A natural worry is that $I(R;X)$ merely measures how badly $\hat\mu$ fits. It does not. Take
the Bayes-optimal location $\hat\mu(x)=\E[Y\given X=x]$, which zeroes the conditional mean
of $R$. Under heteroscedasticity the \emph{conditional law} of the residual still varies, so
$I(R;X)>0$ in general. The gap measures conditional residual scale and shape, and that is
orthogonal to getting the location right.

\section{Discussion}

Coverage error and proper-score performance are two coordinates, both potentially useful. Residual split
conformal moves a report horizontally, changing the set so that coverage error goes to zero
and leaving the density, and hence the score, exactly where it was. Conformal predictive
systems do have a vertical coordinate, but it is capped at the oracle minus $I(R;X)$.
Fitting shape to a proper score moves a report up, and conditioning on $x$ is the only way
to raise it past that cap. A system reported by its coverage alone is a system reported by
one coordinate of a two-coordinate quality.

In the post-hoc setting considered here one should conformalize \emph{last}.
A conformal step applied to a residual score changes the set and leaves the predictive
density, and hence its score, untouched (Remark~\ref{rem:orth}), so a proper score is
improved by estimating and never by certifying. The conformal predictive system is the
exception that proves the point. It does supply a density, and by
Proposition~\ref{prop:gap} it is the best single shape available, but its score stops at the
oracle minus $I(R;X)$.

To close that gap one must condition on $X$. The cheapest way to do so is inside the
conformal step itself, by replacing the plain residual with a normalized score
$s(x,y)=|y-\hat\mu(x)|/\hat\sigma(x)$ \citep{lei2018distribution}, which is exactly the move
that leaves the single-shape class. Whether one then calls the result conformal prediction
or heteroscedastic modeling is a matter of labelling. The sharpness came from $\hat\sigma$.

\paragraph{Scope.} Our strong claims target one object, \emph{post-hoc, marginal, split
conformal prediction with a residual score}. Conditional and shape-adaptive constructions
leave the single-shape class and our result does not apply to them. Examples include
conformalized quantile regression \citep{romano2019conformalized}, Mondrian and binned
conformal predictive systems \citep{bostrom2021mondrian,toccaceli2026crps}, and conformal
training \citep{stutz2022learning}. These can buy sharpness by conditioning on $X$ or
optimizing a proper objective and trading unconditional guarantees, but so can almost every
statistical approach. Where coverage genuinely is the deliverable, as in selective
prediction, shortlisting, or conformal $p$-values for novelty detection
\citep{angelopoulos2021uncertainty,bates2021distribution}, the method is not merely sound
but often ideal. Finally, the exact mutual-information form belongs to the log-score. For
CRPS or the interval score the analogous gap is the regret of the best $X$-blind residual
law under that score, which is not an information quantity.

None of our discussion detracts from the fact that conformal prediction certifies the coverage
 of a set, finite-sample and distribution-free,
under exchangeability. That guarantee is genuine and, for the right objective, exact. Moreover within
the single-shape residual class split conformal prediction is optimal in the sense we have judged it. It
is the choice that may well be the right one for some applications or a reasonable one, either theoretically or pragmatically.  

However coverage answers the question \emph{is
the truth in this region?}, and in forecasting, to name one use, that is frequently {\em not} the question anyone was
asking.  In this note we have gone beyond the usual vague cautions and emphasized
the definable permanent structural concession. Nothing comes without a price, and 
the price for the split conformal coverage guarantee is exactly $I(R;X)$. 

\bibliographystyle{asa-authoryear}
\bibliography{references}

\end{document}